\documentclass[sigconf]{aamas} 

\usepackage[noend]{algpseudocode}
\usepackage{amsthm}
\usepackage{graphicx}
\usepackage{subcaption}
\usepackage{amsmath}
\usepackage{bbm}
\usepackage{multicol}

\usepackage{xcolor}

\usepackage{xcolor}
\usepackage[linesnumbered,ruled,vlined]{algorithm2e}
\SetKwIF{If}{ElseIf}{Else}{if}{}{else if}{else}{end if}%

\SetCommentSty{mycommfont}
\SetKwInput{KwInput}{Input}                
\SetKwInput{KwOutput}{Output}              
\SetKwInput{KwInitialise}{Initialise}      

\algnewcommand{\algorithmicgoto}{\textbf{go to}}%
\algnewcommand{\Goto}[1]{\algorithmicgoto~\ref{#1}}%

\newcommand{\ra}{\ensuremath{\rightarrow}}
\newcommand{\Er}{\mathbf{E}}

\newcommand{\Ac}{\ensuremath{\mathcal{A}}}
\newcommand{\Bc}{\ensuremath{\mathcal{B}}}

\newcommand{\Gc}{\ensuremath{\mathcal{G}}}

\newcommand{\Pc}{\ensuremath{\mathcal{P}}}

\newcommand{\Rbb}{\ensuremath{\mathbb{R}}}
\newcommand{\Ibb}{\ensuremath{\mathbb{I}}}

\newcommand{\lb}{\ensuremath{\left(}}
\newcommand{\rb}{\ensuremath{\right)}}
\newcommand{\lbr}{\ensuremath{\left\{}}
\newcommand{\rbr}{\ensuremath{\right\}}}

\newtheorem{theorem}{Theorem}

\newtheorem{definition}{Definition}

\usepackage{balance} 

\setcopyright{rightsretained}
\acmConference[OptLearnMAS-23]{Appears at the 14th Workshop on Optimization and Learning in Multiagent Systems Workshop (OptLearnMAS 2023). Held as part of the Workshops at the 21st International Conference on Autonomous Agents and Multiagent Systems.}{May-June 2023}{London, England}{Chan, Fioretto, Li, Bistaffa, Kotary (Chairs)}

\copyrightyear{2023}
\acmYear{2023}
\acmDOI{}
\acmPrice{}
\acmISBN{}

\acmSubmissionID{???}

\title[AAMAS-2023 Formatting Instructions]{Bayesian Rationality in Satisfaction Games}

\author{Langford White}
\affiliation{
  \institution{The University of Adelaide}
  \city{Adelaide}
  \country{Australia}}
\email{langford.white@adelaide.edu.au}

\author{Oskar Rynkiewicz}
\affiliation{
  \institution{IMT-Atlantique}
  \city{Brest}
  \country{France}}
\email{oskar.rynkiewicz@gmail.com}

\author{Duong Nguyen}
\affiliation{
  \institution{The University of Adelaide}
  \city{Adelaide}
  \country{Australia}}
\email{duong.nguyen@adelaide.edu.au}

\author{Hung Nguyen}
\affiliation{
  \institution{The University of Adelaide}
  \city{Adelaide}
  \country{Australia}}
\email{hung.nguyen@adelaide.edu.au}

\begin{abstract}
We introduce a new paradigm for game theory -- Bayesian satisfaction. This novel approach is a synthesis of the idea of Bayesian rationality introduced by Aumann, and satisfaction games. The concept of Bayesian rationality for which, in part, Robert Aumann was awarded the Nobel Prize in 2005, is concerned with players in a game acting in their own best interest given a subjective knowledge of the other players' behaviours as represented by a probability distribution. Satisfaction games have emerged in the engineering literature as a way of modelling competitive interactions in resource allocation problems where players seek to attain a specified level of utility, rather than trying to maximise utility. In this paper, we explore the relationship between optimality in Aumann's sense (correlated equilibria), and satisfaction in games. We show that correlated equilibria in a satisfaction game represent stable outcomes in which no player can increase their probability of satisfaction by unilateral deviation from the specified behaviour. Thus, we propose a whole new class of equilibrium outcomes in satisfaction games which include existing notions of equilibria in such games. Iterative algorithms for computing such equilibria based on the existing ideas of regret matching are presented and interpreted within the satisfaction framework. Numerical examples of resource allocation are presented to illustrate the behaviour of these algorithms. A notable feature of these algorithms is that they almost always find equilibrium outcomes whereas existing approaches in satisfaction games may not. 
\end{abstract}


\newcommand{\BibTeX}{\rm B\kern-.05em{\sc i\kern-.025em b}\kern-.08em\TeX}

\begin{document}


\pagestyle{fancy}
\fancyhead{}


\maketitle 


\section{Introduction}

\subsection{What are Satisfaction Games}

A {\it satisfaction game} (SG) is a strategic game in which players compete in a manner by which they each attempt to meet a satisfaction constraint. This is in contrast to the usual normal form game where players are trying to maximize their individual utilities. As in the normal form game \cite{Feud_Tir_2000}, there is a finite set of $N \geq 2$ players which we label as the set $\Pc = \{1, \ldots, N\}$. Each player $i \in \Pc$ has its associated finite set of {\it actions} $\Ac_i$. The set of {\it action profiles} of the game is the Cartesian product $\Ac = \Ac_1 \times \Ac_2 \times \cdots \Ac_N$. Following the usual abuse of notation, we'll often write an action profile $a \in \Ac$ as $a = (a_i,a_{-i})$, where $a_i \in \Ac_i$ and $a_{-i} \in  \Ac_{-i}$, where $\Ac_{-i} = \Ac_1 \times \cdots \times \Ac_{i-1} \times \Ac_{i+1} \times \cdots \times \Ac_N$. Denote by $\Delta(\Ac)$ the set of all probability mass functions (pmf) over $\Ac$, and similarly, for each $i$, $\Delta(\Ac_i)$ as the set of pmf over $\Ac_i$, and $\Delta(\Ac_{-i})$ as the set of pmf over $\Ac_{-i}$. \\

Rather than utility functions, a SG has a set of {\it correspondence functions}, one for each player $i$, $f_i : \Ac_{-i} \ra 2^{\Ac_i}$.  
A SG is then represented by the tuple $\Gc = \lb \Pc, \{\Ac_i, i \in \Pc\}, \{f_i, i \in \Pc\}\rb$. In strategic games such as a SG, we are generally interested in equilibrium behavior. Equilibria are action profiles whereby there is no incentive for a player to deviate from their action, given that the actions of the other players are unchanged. A (pure) {\it satisfaction equilibrium} (SE) of the SG $\Gc$ is an action profile whereby all players are satisfied, i.e. $a_i \in f_i(a_{-i})$ for all $i \in \Pc$. The concept of SG and SE was originally introduced in \cite{ross2006satisfaction} but in the less general setting of resource allocation. As argued in \cite{goonewardena2017generalized}, the concept of SE is somewhat restrictive, and they define the idea of a {\it generalized} SE (GSE) as an action profile whereby some players are satisfied and others are not, but none of the unsatisfied players can become satisfied by unilateral deviation from their specified actions. The above discussion concerns pure strategies; we'll discuss mixed strategies and SE/GSE subsequently.  

\subsection{Why are Satisfaction Games Important ?}

In many applications, players (agents or users) do not necessarily seek to maximize their individual utility, but rather may only wish to attain a certain level of utility. Most of the engineering literature we are aware of that uses the idea of SG, focuses on applications of resource allocation in wireless communications systems. For example, \cite{perlaza2010satisfaction,perlaza2011quality} consider power control for multiple users accessing an interference-limited channel. In \cite{goonewardena2017generalized,FASOULAKIS2019135}, several other aspects of management of a wireless network including energy management, admission control and channel allocation as well as power control are considered within a SG framework. Satisfaction ideas are also applied in \cite{arani2017minimizing} to address the problem of excessive switchings between different radio access technologies in heterogeneous wireless networks. Applications of SG to underwater communications are considered in \cite{pottier2018quality} and in linear control systems~\cite{8723508}.\\ 

In each of these applications, satisfaction is usually defined as a player attaining a specified level of utility (e.g. SNR, throughput) which is only one kind of SG. In these cases, utility functions are defined along with a set of minimum acceptable levels of utility for each player. If we denote the utility functions by $u_i :\Ac \ra \Rbb$ and the minimum required level of utility for player $i$ as $\Gamma_i$, then the correspondence functions are defined by $f_i \lb a_{-i} \rb = \lbr a_i \in \Ac_i : u_i(a_i,a_{-i}) \geq \Gamma_i \rbr$. An SE arises when all players achieve their desired minimum utility. Although most applications do have utility functions, it is of interest to retain the more general framework based on correspondence functions. \\ 

Another important aspect of SG is that methods for finding SE/GSE based on repetition are likely
to converge more quickly than those for normal-form games. This is because attaining satisfaction, or
at least, maximizing the probability of satisfaction, should require less effort than maximizing utilities.
We will demonstrate this property in the examples presented in sec. \ref{sec:examples}.

\subsection{Relevant Background Work\label{ssec:background}}

The concept of a SE is due to Ross and Chaib-draa \cite{ross2006satisfaction}, who introduced the idea in a more restrictive form where satisfaction is defined as attaining a specified level of utility. This definition was generalized by Perlaza {\it et al} in \cite{perlaza2010satisfaction,perlaza2011quality} using correspondence functions, without the requirement of any notion of utility. The formulation of \cite{ross2006satisfaction} is framed in the terminology of {\it repeated play}, and, in particular, considers equilibrium behavior whereby all players are satisfied and thus have no incentive to change their strategies. However, although repetition may provide a basis for computing SE, it is not necessary to invoke repetition to properly define the concept of SE. Indeed \cite{perlaza2010satisfaction,perlaza2011quality} formalize the idea of a game in satisfaction form, rather than the usual normal form which involves utility (or payoff) functions. The concept of an SE follows directly from the satisfaction form and does not involve introducing ideas of repetition. Satisfaction is defined solely in terms of the properties of the corresponding functions of the SG. \\ 

In \cite{goonewardena2017generalized}, the concept of a {\it generalized} SE (GSE) was introduced. A GSE is an equilibrium notion for a SG whereby not all players are satisfied, but those who are not satisfied have no incentive to unilaterally deviate from their specified actions, i.e. they are not able to become satisfied by unilateral action. More specifically, $a \in \Ac$ is a (pure) GSE if there is a disjoint partition of the set of players $\Pc = \Pc_u \cup \Pc_s$ where, for all $i \in \Pc_s$, $a_i \in f_i(a_{-i})$, and for all $i \in \Pc_u$, $f_i(a_{-i}) = \emptyset$. Clearly, every SE is a GSE (then $\Pc_s = \Pc$), but not conversely. This paper will be primarily concerned with GSE. \\

In understanding the concepts of SE and GSE, reference is made in the literature to the relationships between SE/GSE and Nash equilibria (NE). If we have a SG, then we can define utility functions $u_i : \Ac \ra \{0,1\}$, by $u_i(a)  = \Ibb \lbr a_i \in f_i(a_{-i}) \rbr$ \footnote{We use the notation $\Ibb\{p\}$ to give the value of one if the proposition $p$ is true and zero otherwise.}, i.e. a player receives a utility of one if it is satisfied under $a$, and zero otherwise. We will refer to these utility functions as the {\it natural utility functions} associated with an SG, and the normal form game $\Gc^\prime = \lb \Pc, \{\Ac_i, i \in \Pc\}, \{u_i, i \in \Pc\} \rb$ as the {\it natural normal form game} associated with a SG. In \cite{perlaza2011quality}, it's shown that any pure SE for $\Gc$ is a pure NE for $\Gc^\prime$ but the converse is not true in general. In \cite{goonewardena2017generalized}, it's shown that the pure GSE of $\Gc$ are identical to the pure NE of $\Gc^\prime$. This is because the unsatisfied players have no incentive to unilaterally deviate from the specified action profile even though they don't obtain the utility of one. Significantly, there is no general proof about the existence of GSE \cite{goonewardena2017generalized}. \\

In the case of mixed strategies, the situation is less clear. In \cite{goonewardena2017generalized}, it's claimed that the mixed GSE of $\Gc$ don't, in general, coincide with the mixed NE of $\Gc^\prime$. A counter-example using the constraint satisfaction correspondences is provided but there's no characterization of the class of correspondences for which these equilibria sets do coincide. In the mixed strategy case, it follows that expected utility, the utility now regarded as a random variable, under some $\pi \in \Delta(\Ac)$ is the probability of satisfaction under $\pi$, i.e. $\Er_\pi \lbr U_i \rbr = \Pr_\pi \lbr A_i \in f_i(A_{-i}) \rbr$.\footnote{We use upper case $A$ to denote a random variable taking values $a \in \Ac$, and upper case $U_i$ to denote the random variable taking values $u_i$.} The event $A_i \in f_i(A_{-i})$ means the disjoint union $\cup_{a_i \in \Ac_i} \ \{a_i\} \times f^{-1}(a_{i}) $ and $P_\pi$ is the probability function defined over $\Ac$ generated by $\pi$. \\

Naturally, methods for finding SE/GSE are of great significance. Unfortunately, in general, finding SE/GSE is an NP-hard problem (see \cite{goonewardena2017generalized} for more specific statements about complexity). Several algorithms for finding SE have appeared in the literature. In \cite{ross2006satisfaction} and \cite{perlaza2010satisfaction,perlaza2011quality} a procedure called {\it Pure Satisfaction Equilibrium Learning} (PSEL) is considered. Although the algorithm is the same in both \cite{ross2006satisfaction} and \cite{perlaza2010satisfaction,perlaza2011quality}, the latter is more general due to the more general definition of SE. PSEL is a procedure by which the SG is repeated in order to find an SE. The idea is that satisfied players keep playing the same (pure) strategy whilst the remaining players play a mixed strategy (exploration). Under certain sufficient conditions, the most important being that an SE exists, and that for every $i \in \Pc$, and all $a_{-i} \in \Ac_{-i}$, the set $f_i \lb a_{-i} \rb$ is non-empty, the PSEL algorithm converges to a CE provided all ``exploration'' probabilities are non-zero. A simple choice \cite{ross2006satisfaction} is to take a uniform distribution for each unsatisfied player over all its actions. However, we find the statements of convergence in both \cite{ross2006satisfaction} and \cite{perlaza2010satisfaction,perlaza2011quality} somewhat unclear. In particular, \cite{ross2006satisfaction} suggests that PSEL may converge but not to an equilibrium. We have observed this behavior in simulations. \\ 

Furthermore, it seems that in \cite{perlaza2011quality} (Proposition 10), that PSEL converges in a finite time almost surely. However, \cite{perlaza2011quality} then introduces the concept of a {\it clipping action} which is referred to in \cite{ross2006satisfaction} as a {\it satisfying strategy}. A player $i \in \Pc$ has a clipping action, $a_i^* \in \Ac_i$ if for every $a_{-i} \in \Ac_{-i}$, it holds that $a_i^* \in f_i \lb a_{-i} \rb$, i.e. player $i$ is always satisfied if it plays $a_i^*$ irrespective of what the other players do. Then \cite{perlaza2011quality} (Proposition 12) argues that if there's another player which cannot be satisfied if player $i$ plays its clipping action, PSEL ``doesn't converge to a SE with strictly positive probability''. Then it's suggested that this problem might be rectified by allowing satisfied players to change their strategies. This idea is certainly applied in our proposed approach. In \cite{perlaza2011quality}, a similar approach is used (Algorithm 1) with some additional refinements. This approach is termed a {\it satisfaction response} algorithm (SRA). A player among subset of unsatisfied players (possible asynchronously) chooses an action for which it becomes satisfied. However a general convergence proof is not provided rather one for some specific SG having a kind of ordering property on its action space. The SRA algorithm is applied with some modifications in \cite{pottier2018quality} where convergence is proven. We don't address this approach here, apart from remarking that in this algorithm, satisfied players don't change their actions. In terms of finding {\it mixed} SE/GSE, we concentrate on \cite{perlaza2011quality} because its definition of mixed SE/GSE is the most general. Again their algorithm (Algorithm 2) requires a kind of ordering property, this time on the space of pmfs over actions generated by the assumed ordering over actions. We don't address this approach here because in our approach, we don't wish to make any of these kinds of restrictive assumptions. \\

In the literature, there is also considerable attention to the so-called {\it efficient} SE (ESE). In such a framework, there is a cost associated with each player's actions, however these costs don't depend on the others' actions. The idea of ESE was first proposed in \cite{perlaza2011quality} in the context of equilibrium selection. However, naturally, there is a notion of ordering on actions imposed by the costs. Subsequently in \cite{el2017fully}, the authors considered in more detail the notion and properties of ESE and other concepts of equilibria associated with fading channels. Again, in these examples, orderings on actions was an important assumption. 

\subsection{Our Contributions and Significance}
In this paper, we are concerned with SG $\Gc$ regarded in the natural normal form $\Gc^\prime$. We study the relationships between GSE for $\Gc$ and {\it correlated equilibria} (CE) \cite{aumann1987correlated} for $\Gc^\prime$. CE is an expression of {\it Bayesian rationality}. A player is said to be rational if it acts in its own best interest given its information. In a Bayesian setting, the information each player has is represented by a pmf over actions. We firstly show that CE is a natural optimality concept for SG in the sense that if a mixed strategy profile $\pi$ is a CE for $\Gc^\prime$, even if $\pi$ is not a mixed GSE, it yields the highest probability of satisfaction for all players. We will make this concept more precise in sec. \ref{sec:CE_sat_games}. \\

There is a well-known algorithm for finding CE called Regret Matching (RM) \cite{hart2000simple}. Although the convergence properties of RM are unsatisfactory in that only {\it almost sure approachability} to the set of CE is guaranteed, if RM (in its standard form) converges to a point in CE, it is necessarily a pure NE. However all pure NE for $\Gc^\prime$ are necessarily pure GSE. So if $\Gc$ has at least one pure GSE, then RM will tend to find one. In sec. \ref{sec:examples}, we provide examples to illustrate the convergence properties of RM for the natural SG. We compare RM with the PSEL and SRA algorithms and show that although convergence properties depend on the kind of SG considered, RM (almost) always attains a GSE, whereas PSEL may fail to converge to a GSE or even an equilibrium. \\
 
The information structure of the SG is important when considering various algorithms for finding GSE. For example, in both PSEL and SRA, at the completion of each stage of the game, all players need to know if they are satisfied or not. Each player can deduce this given its own correspondence function and the actions of all other players. Alternatively, there may be a ``coordinator'' which knows all correspondence functions and all actions. The coordinator can then signal each player as to whether it is satisfied or not. This is like the ``one-bit feedback'' idea of \cite{pottier2018quality}.\\ 

In PSEL each player needs only to know whether it is satisfied or not; no additional information is needed. In this sense, PSEL is decentralized. For SRA, in order for an unsatisfied player to choose an action for which it becomes satisfied, it needs to know the actions of all the other players as well as its own correspondence function. Alternatively, a coordinator can specify to each unsatisfied player what actions it can take in order to be satisfied (i.e. specify a satisfaction response in the language of \cite{goonewardena2017generalized}). In the RM algorithm for the natural SG, each player needs to reason about whether it would be satisfied choosing all its actions given the actions of the other players. So it requires not only these other actions but also its own correspondence function. Thus RM requires similar information to SRA. These algorithms are thus not decentralized in this respect. \\

A reinforcement learning (RL) approach to RM (RMRL) has been proposed \cite{hart2001reinforcement} which has the property that each player needs not know the actions of the other players but can estimate regrets based on its own past actions and payoffs. In the context of the natural SG, with RMRL each player needs only to know whether it was satisfied or not for each action it played in the past. Thus RMRL requires the same information as PSEL but less than SRA. We will assess the performance of the RMRL algorithm in sec. \ref{sec:examples}.

\subsection{Paper Layout}
Following this extensive introduction (which we believe necessary to properly define the setting), sec. \ref{sec:CE_sat_games} considers the nature of CE for a SG and proposes a new definition for a mixed GSE. It is then shown that all such GSE are CE but not conversely. Thus the set of CE form new equilibria which have not previously been considered for SG. These equilibria possess an optimality property that no player can increase its probability of satisfaction by deviating from the equilibrium. In sec. \ref{sec:RM}, two algorithms (RM and RMRL) based on the notion of regret are introduced. These are specializations of existing algorithms to the SG case. In sec. \ref{sec:examples}, we consider two practical examples of resource allocation to examine the performance of the regret-based algorithms. The paper concludes with suggestions for ongoing related research.


\section{Correlated Equilibria for Satisfaction Games\label{sec:CE_sat_games}}

In this section, we explore the relationship between CE and SE/GSE for satisfaction games (where the binary utility is specified by the correspondence functions). In particular, we characterize CE as those mixed strategies for which each user cannot improve its probability of satisfaction (in a precise Bayesian sense) by deviating from the specified action. We need to provide an alternative definition of mixed SE/GSE to those currently used in the literature due to the specific information structure under which CE is defined. This definition is shown to be well-defined in a probabilistic sense by use of the correspondence functions of the SG. \\

Consider a general repeated (finite) strategic game in normal form $\lb \Pc, \{\Ac_i, i \in \Pc\}, \{u_i, i \in \Pc\} \rb$. A pmf $\pi \in \Delta(\Ac)$ is a {\it correlated equilibrium} (CE) \cite{aumann1987correlated} if for every player $i$, and every pair of actions $a_i, a_i^\prime \in \Ac_i$, it holds that
\begin{align}
0 & \geq \sum_{a_{-i}\in\Ac_{-i}} \pi(a_{-i}|a_i) \lb u_i(a_i^\prime,a_{-i}) - u_i(a_i,a_{-i}) \rb \nonumber \\ &  = \Er_{\pi|a_i} \lbr u_i(a_i^\prime,a_{-i}) \rbr - \Er_{\pi|a_i} \lbr u_i(a_i,a_{-i}) \rbr  = R_i^\pi(a_i,a_i^\prime) \ ,
\label{eq:CE1}
\end{align}
where $\pi(.|a_i)$ is the conditional pmf over $a_{-i} \in \Ac_{-i}$ given $a_i$, and $\Er_{\pi|a_i}$ denotes expectation with respect to $\pi(.|a_i)$. The quantities $R_i^\pi(a_i,a_i^\prime)$ are known as the {\it expected regrets} under $\pi$ for player $i$ not playing $a_i^\prime$ instead of playing $a_i$. Thus if for a given $a_i$, $R_i^\pi(a_i,a_i^\prime) > 0$, player $i$ would obtain a higher average reward (under $\pi$) if it substituted $a_i^\prime$ for $a_i$ given that it was expected to play $a_i$. \\

CE is a natural expression of Bayesian rationality \cite{aumann1987correlated}, although the set of linear inequalities in \eqref{eq:CE1} are a little tricky to interpret. One interpretation is to think of there being a ``co-ordinator'' of the game that makes recommendations on actions to the players. The co-ordinator draws from a mixed action profile $\pi$ and sends the realized actions $a_i$ individually to each player $i$. Each player knows $\pi$. If $\pi$ is a CE, then there is no incentive for any player to deviate from the recommended action $a_i$ given $\pi$. Actions $a_i$ which have zero probability need not be considered \cite{hart1989existence}. \\

Consider now the natural game $\Gc^\prime$ for a SE game $\Gc$, then from \eqref{eq:CE1}, a CE $\pi$ thus satisfies, for all players $i$, and all actions $a_i, a_i^\prime \in \Ac_i$,
\begin{align}
0 & \geq \Er_{\pi|a_i} \lbr \Ibb \lbr a_i^\prime \in f_i(a_{-i}) \rbr -  \Ibb \lbr a_i \in f_i(a_{-i}) \rbr \rbr \nonumber \\
& = \Pr_{\pi|a_i} \lbr a_i^\prime \in f_i(a_{-i}) \rbr -  \Pr_{\pi|a_i} \lbr a_i \in f_i(a_{-i}) \rbr \ ,
\label{eq:CE2}
\end{align}
where $\Pr_{\pi|a_i}$ are the probability functions defined over $\Ac_{-i}$ induced by the conditional pmf $\pi(.|a_i)$.
Thus under a CE of $\Gc^\prime$, no player can unilaterally increase its probability of ``satisfaction'' (by choice of its own action) by deviating from the specified action $a_i$ given the other players play according to $\pi$.\\ 

The concepts of mixed SE/GSE presented in \cite{perlaza2011quality} and \cite{goonewardena2017generalized} are different as pointed out in \cite{goonewardena2017generalized}. The mixed SE of \cite{perlaza2011quality} are mixtures of pure SE, whilst the mixed GSE of \cite{goonewardena2017generalized} are not, in general, mixtures of pure GSE. In the context of linking mixed GSE of $\Gc$ and CE of $\Gc^\prime$, we think of mixed GSE in the following way.
 
 \begin{definition}\label{defn1}
 A joint distribution $\pi \in \Delta(\Ac)$ is a mixed GSE for $\Gc$ if there is a partition $\Pc = \Pc_s \cup \Pc_u$ of the players such that (i) for every player $i \in \Pc_s$, $\Pr_{\pi} \lbr a_i \in f_i(a_{-i}) \rbr =1$, and (ii) for every player $i \in \Pc_u$, $\Pr_{\pi} \lbr a_i \in f_i(a_{-i}) \rbr =0$. 
 \end{definition}
 
\underline{Comments}
\begin{enumerate}
\item Let $\pi_i(a_i)$ be the marginal distribution of choosing $a_i$ under the joint distribution $\pi$. The probability of satisfaction for any player $i \in \Pc$ under joint mixed action profile $\pi \in \Delta(\Ac)$ is given by
\begin{align}
&\Pr_\pi \lbr a_i \in f_i(a_{-i}) \rbr  = \sum_{a \in \Ac} \ \pi(a) \,  \Ibb \lbr a_i \in f_i(a_{-i}) \rbr \nonumber \\
& = \sum_{a_i \in \Ac_i} \pi_i(a_i) \ \sum_{a_{-i} \in \Ac_{-i}} \pi(a_{-i}|a_i) \, \Ibb \lbr a_i \in f_i(a_{-i}) \rbr \nonumber \\
& = \sum_{a_i \in \Ac_i} \pi_i(a_i) \ \Pr_{\pi|a_i} \lbr a_i \in f_i(a_{-i})  \rbr \ ,
\label{eq:P_sat}
\end{align}
where only terms with $\pi_i(a_i) > 0$ are included in the outer sum (so that, in particular, the conditional pmfs $\pi(.|a_i)$ are defined). Suppose $i \in \Pc_s$, then, by defn. \ref{defn1}, the l.h.s. of \eqref{eq:P_sat} takes the value one, which implies that each of the terms $\Pr_{\pi|a_i} \lbr a_i \in f_i(a_{-i})  \rbr = 1$. Thus there is no incentive for player $i$ to unilaterally deviate from any recommended action $a_i$ (having non-zero probability). Similarly, suppose $i \in \Pc_u$, then the l.h.s. of \eqref{eq:P_sat} takes the value zero, which implies that each of the terms $\Pr_{\pi|a_i} \lbr a_i \in f_i(a_{-i}) \rbr = 0$. Thus player $i$ cannot become satisfied under any recommended action $a_i$ (having non-zero probability). Indeed, no such player can achieve a non-zero probability of satisfaction by choice of any such action. 

\item In the case of pure strategies, definition \ref{defn1} is consistent with the existing definition of pure GSE \cite{goonewardena2017generalized}. 

\item The definition of mixed GSE given in \cite{goonewardena2017generalized} is given in terms of correspondence functions defined over pmfs. A $\pi \in \Delta(\Ac)$ is a mixed GSE if either $\pi_i \in f_i(\pi_{-i})$ (player $i$ satisfied) or $f_i(\pi_{-i}) = \emptyset$ (player $i$ unsatisfied). Under this scenario, each satisfied player $i$ chooses its action according to the marginal distribution $\pi_i$. There's no manner by which unsatisfied users can choose an action to become satisfied. Indeed they don't even ``play''. Also actions are chosen by satisfied players {\it independently}. This is in contrast with the CE setting and definition~\ref{defn1} where an action profile $a$ is chosen jointly according to $\pi$ and each player plays the specified $a_i$. 
\end{enumerate}
 
 Our main result is :
\begin{theorem}\label{thrm1}
Any mixed GSE for $\Gc$ is a CE for $\Gc^\prime$. 
\end{theorem}
\begin{proof}
Suppose $\pi \in \Delta(\Ac)$ is a GSE for $\Gc$. Suppose a player $i$ is satisfied under $\pi$, then $\Pr_{\pi|a_i} \lbr f_i^{-1}(a_i) \rbr = 1$ for all actions $a_i$ with $\pi_i(a_i) > 0$. Thus, for any such $a_i$,
\begin{align*}
1 & = \Pr_{\pi|a_i} \lbr a_i \in f_i(a_{-i}) \rbr \geq \Pr_{\pi|a_i} \lbr a_i^\prime   \in f_i(a_{-i})\rbr \ ,
\end{align*}
for any action $a_i^\prime \in \Ac_i$. Thus \eqref{eq:CE2} holds. Similarly, for a player $i \in \Pc_u$, $\Pr_{\pi|a_i} \lbr a_i \in f_i(a_{-i}) \rbr = 0$ for all actions $a_i$ with $\pi_i(a_i) > 0$, so \eqref{eq:CE2} holds trivially. 
\end{proof}

\underline{Comments}
\begin{enumerate}
\item The converse to theorem \ref{thrm1} is not in general true, since a CE may be an equilibrium mixed strategy profile where players are neither satisfied nor unsatisfied.
\item Any pure CE is a Nash equilibrium (NE) for $\Gc^\prime$, and as such is a GSE for $\Gc$. However CE contains all mixed GSE as well as all mixed NE but as pointed out in \cite{goonewardena2017generalized}, these latter two sets need not coincide.
\item In this framework, CE is an expression of Bayesian satisfaction. No player has an incentive to deviate from its recommended action in the sense that it cannot increase its probability of satisfaction under the assumption that all players also play their recommended actions.
\end{enumerate}

\section{Regret Matching Based Algorithms\label{sec:RM}}
\subsection{Regret Matching for Satisfaction Games}
Regret Matching (RM) \cite{hart2000simple} is an iterative procedure for computing CE in a general normal-form game. The idea is that each player computes a matrix of estimated ``regrets'' which it updates each iteration. These regrets specify the (mixed) action that each player chooses at the next iteration. Convergence of the joint empirical distribution of action profiles (in a precise sense) to the set of CE is guaranteed almost surely. More specifically, for each iteration $t \geq 1$, each player $i$ computes
\begin{align}
R_i^t(a_i,a_i^\prime) & = \frac{1}{t} \ \sum_{n=1}^t \mathbb{I}\lbr s_i^n = a_i \rbr \ \lb u_i(a_i^\prime,s_{-i}^n ) - u_i(s^n) \rb \ ,
\label{eq:Regrets}
\end{align}
for $a_i, a_i^\prime \in \Ac_i$, where $s^n \in \Ac$ is the joint action played at iteration $n\geq 1$. The quantity $u_i(s^n)$ is the payoff player $i$ received at iteration $n$ and is known to player $i$. Thus at each iteration that player $i$ actually played action $a_i$, it reasons about what payoff it would have received if instead, it had played any of its other actions $a_i^\prime$ when all other players played the same actions. Positive values indicate a ``regret'' for player $i$ not having played $a_i^\prime$ rather than $a_i$ at each of those iterations.\\ 

At the next iteration, each player independently chooses its next action according to the probabilities \footnote{The notation $[x]_+ = \max\{x,0\}$.}
\begin{align}
\Pr \lbr s_i^{t+1} = a_i^\prime | s_i^t = a_i \rbr & = c_i^{t+1} \ \left[ R_i^t(a_i,a_i^\prime) \right]_+ \ ,
\label{eq:RM}
\end{align}
for all $a_i^\prime \neq a_i$, and the quantities $c_i^{t+1}$ are chosen so that the probability $\Pr \lbr s_i^{t+1} = a_i | s_i^t = a_i \rbr$ of playing the same action at the next iteration is positive. It's shown in \cite{hart2000simple} that the empirical distribution $\hat{\pi}_t(a) = | \{ s_n = a : n \leq t \} | / t$ converges to the set of CE for the game as $t \ra \infty$, where convergence means that the distance between $\hat{\pi}_t$ and the set of CE goes to zero almost surely. As also discussed in \cite{hart2005adaptive}, the convergence of RM can be problematic - if convergence to a point does occur than that point is a pure NE. However, the regrets generally decrease slowly to zero, and although the empirical distributions of plays approach the set of CE, there is often a switching between different outcomes.\\ 

For the natural SG, the regrets in \eqref{eq:Regrets} have the form
\begin{align}
R_i^t(a_i,a_i^\prime) = & \frac{1}{t} \ \sum_{n=1,s_i^n=a_i}^t   \mathbb{I}\lbr (a_i^\prime \in f_i(s_{-i}^n) \rbr \nonumber \\  
& - \frac{1}{t} \ \sum_{n=1,s_i^n=a_i}^t \mathbb{I}\lbr (a_i \in f_i(s_{-i}^n) )\rbr  \ .
\label{eq:Regrets_SG}
\end{align}
Examining the summand in \eqref{eq:Regrets_SG}, we see that at each iteration where player $i$ played action $a_i$, the second term is one if player $i$ is satisfied, and zero otherwise. This is the known ``payoff'' for player $i$ - whether it is satisfied at iteration $n$ or not. The first term will yield one if player $i$ would be satisfied replacing $a_i$ (the action it did play) with another $a_i^\prime$. Thus the summand can take values $-1$ (if it's currently satisfied playing $a_i$ and would become unsatisfied playing $a_i^\prime$ instead) ; +1 (if it's currently unsatisfied playing $a_i$ and would become satisfied playing $a_i^\prime$ instead) and zero for the other two possibilities. So positive regrets indicate a tendency to switch to actions which increase the probability of satisfaction. Notice that in order to reason about its satisfaction (or otherwise) is deviating from its actual play at each iteration $n$, each player needs to know the actions $s_{-i}^n$ played by the others. Thus RM is not a decentralized algorithm. 



\subsection{Reinforcement Learning Approach to Regret Matching for Satisfaction Games}
In order to yield a fully decentralized algorithm for computing CE, \cite{hart2001reinforcement} proposed a RM type procedure whereby estimated regrets can be computed using only a player's past history. The algorithm called Regret Matching with Reinforcement Learning (RMRL) replaces the regret calculations in a normal form game as in eqn. \eqref{eq:Regrets}, with
\begin{align}
\hat{R}_i^t(a_i,a_i^\prime) = & \frac{1}{t} \ \sum_{n=1}^t  \mathbb{I}\lbr s_i^n = a_i^\prime \rbr \, u_i(s^n ) \, \frac{ \Pr \lbr s_i^n = a_i \rbr}{\Pr \lbr s_i^n = a_i^\prime \rbr} \nonumber \\  
& - \frac{1}{t} \ \sum_{n=1}^t  \mathbb{I}\lbr s_i^n = a_i \rbr \, u_i(s^n) \ . 
\label{eq:Est_Regrets}
\end{align}

For SG, we use the binary utilities as in eqn. \eqref{eq:Regrets_SG}. 
In \eqref{eq:Est_Regrets} all quantities involved are known to player $i$ - each player needs only know if it is satisfied at each iteration or not. The probabilities appearing in \eqref{eq:Est_Regrets} are the empirical probabilities of play for player $i$ at stage $n$. In RMRL, some ``tremble'' is required to be added to the play probabilities in \eqref{eq:RM} in order to obtain sufficient ``exploration'' in the learning process. For RMRL, convergence of the empirical distributions on plays converges to the set of CE provided the level of the tremble decreases sufficiently as $t \ra \infty$ (see \cite{hart2001reinforcement} for details.) We also found in our simulations that it is desirable to ``switch off'' trembles after some defined number of iterations, however the full understanding of how to select the parameters for RMRL is still an important subject for study, and is application dependent.




\section{Numerical Examples\label{sec:examples}}

%

\subsection{Resource Allocation Games}

We consider a resource allocation game where there are $N$ agents competing for $M$ resources. Each agent $i$ has a minimum resource (satisfaction) level $\Gamma_i$, and each resource $j$ has a capacity of $C_j$ units. The action space of each agent is $\Ac_i = \{1, \ldots, M\}$ where action $j$ for agent $i$ means means that $i$ connects to resource $j$. If multiple agents connect to a given resource, then they are ``served'' in a randomly (uniform) chosen order. The resource serves each agent until it can no longer supply demand. An agent is satisfied if it attains its minimum resource level $\Gamma_i$ otherwise it is unsatisfied. There is no ordering on actions, and the agents do not know the resource levels $C_j$. 

In these simulations, we used $N=20$ agents and $M=10$ resources. The satisfaction levels $\Gamma_i$ for each agent were uniformly chosen on $[1,5]$. The resource capacities were uniformly distributed on $[0,10]$. The number of iterations $T = 5000$, and 250 independent realizations were averaged.
The average resource capacity = 5.87 and average demand = 3.57. 

\begin{table}[!t]
	\small 
	\begin{center}
		\begin{tabular}{c|cccc}
			Algorithm & Prob (Equil) & Prob (Sat) & Avg. Utility & Alloc. Eff. \\ \hline
			PSEL (unif) & 0.165 & 0.660 & 0.610 & 0.742 \\
			PSEL (reinf) & 0.095 & 0.630 & 0.576 & 0.700\\
			PSRA (1) & 0.110 & 0.632 & 0.581 & 0.705 \\
			PSRA (2) & 0.110 & 0.638 & 0.585 & 0.711 \\
			PSRA (10) & 0.110 & 0.647 & 0.600 & 0.724 \\
			RMRL & 0.730 & 0.750 & 0.705 & 0.858\\
			RM & 0.960 & 0.747 & 0.681 & 0.827 \\ 
			\hline
		\end{tabular}
	\end{center}
	\caption{Table shows the performance of PSEL, PSRA, RM and RMRL algorithms for a resource allocation game with 25000 iterations. Here there are $N=20$ agents accessing $N=10$ resources. Average resource capacity = 5.87 and average demand = 3.57.\label{table1}}
\end{table}

We compared PSEL with uniform, and ``reinforced'' exploration \cite{perlaza2011quality}, PSRA (with 1, 2, and 10 maximum unsatisfied agents accessing resources at each iteration), RM and RMRL. Table \ref{table1} shows the probability an equilibrium was obtained, the probability of satisfaction, the average utility achieved and the average resource allocation efficiency. Although RM achieves higher levels of satisfaction, utility and resource efficiency, the most striking aspect of these results is that RM (almost) always achieves an equilibrium (i.e. stable) outcome whilst the other algorithms fail to do so most of the time. RMRL also performs well in comparison with PSEL and PSRA, although the selection of the adaptation parameters requires some care. 

In Table~\ref{table2} we repeat for the case where there are $N=30$ agents, but where all other parameters remain the same. The average resource capacity = 3.98 and the average demand = 3.50 (a more highly loaded case). Again, RM and RMRL are superior in terms of the probability of attaining an equilibrium and slightly better in terms of the other criteria.

\subsection{Satisfaction in RAT Selection}
In this example, we demonstrate the performance of RM and RMRL on a radio access technology (RAT) selection game~\cite{nguyen2017reinforcement,nguyen2018evaluating}, where players (network users) seek to achieve a certain minimum level of throughput. For simplicity, we assume that all the users have the same level of satisfaction (i.e. a common satisfaction threshold). We follow the same throughput models in~\cite{nguyen2017reinforcement} to simulate a user-specific throughput of a user associated with a RAT (either a WiFi or a LTE base station). In our simulation, there are 100 users located in the overlap coverage area of 10 base stations (BSs), in which half of them are WiFi and the rest are LTE. \\

To compute the RM algorithm, we assume that every BS sends a binary signal to all users, including the users that are currently not associated with it, to inform whether a user will become satisfied (i.e. a utility of one if satisfied) or unsatisfied (i.e. a utility of zero if not satisfied) by selecting it. As a result of this, the RAT selection problem can be considered as a satisfaction game with binary utility. A user running RM needs to know the actions of the other users in order to compute its regrets. A user running RMRL only needs to know if it is currently satisfied or not with its serving BSs, and the historical experiences for all the actions taken in the past. It requires to know neither its own utility function nor the actions taken by all the other players. Thus, the RMRL is a fully decentralized algorithm.\\

Figures~\ref{fig:1.5Mbps} --~\ref{fig:GSE} illustrate several examples of convergence toward a SE/GSE using either RM or RMRL algorithm on the RAT selection game. As can be seen from the graphs, RM converges very quickly towards a pure SE for a low satisfaction threshold of per-user throughput (i.e. 1.5 Mbps). RMRL -- a fully distributed version of RM has similar performance but takes longer time for the game to converge as it uses less information compared to that of RM. Also, increasing the satisfaction threshold requires longer convergence time for both RM and RMRL as evidenced in fig. \ref{fig:2Mbps}. \\

Figure~\ref{fig:RM_GSE_2-2} demonstrates the convergence towards a pure GSE using RM, in which 52 of the 100 users are satisfied with their serving BSs while the remainder of users are unsatisfied but has chosen to remain with their existing actions as they could not get better throughput by switching to any other BS. Figure~\ref{fig:RM_MGSE_2-8} demonstrates the convergence of using RM towards a mixed GSE instead of a pure GSE as illustrated in Figure~\ref{fig:RM_GSE_2-2}. Under this convergence, some of the unsatisfied users stay with a particular BS from their possible actions, while the other unsatisfied users play a mixed strategy according to a converged probability distribution to pick their next actions as the iterations progress.  

\begin{table}[!t]
	\small 
	\begin{center}
		\begin{tabular}{c|cccc}
			Algorithm & Prob (Equil) & Prob (Sat) & Avg. Utility & Alloc. Eff. \\ \hline
			PSEL (unif)  & 0.360 & 0.364 & 0.303 & 0.793 \\
			PSEL (reinf) & 0.120 & 0.348 & 0.291 & 0.772  \\
			PSRA (1) & 0.160 & 0.355 & 0.300 & 0.790    \\
			PSRA (2) & 0.360 & 0.360 & 0.300 & 0.803   \\
			PSRA (10) &  0.360 &  0.363& 0.302  & 0.784 \\
			RMRL & 0.520 & 0.384 & 0.318 & 0.836  \\
			RM &  0.760 &  0.408 &  0.322 & 0.850 \\ 
			\hline
		\end{tabular}
	\end{center}
	\caption{Table shows the performance of PSEL, PSRA, RM and RMRL algorithms for a resource allocation game with 25000 iterations. Here there are $N=30$ agents accessing $N=10$ resources. Average resource capacity = 3.98  and average demand = 3.50. \label{table2}}
\end{table}

\begin{figure*}[!t] 
        \centering
        \begin{subfigure}[b]{0.39\linewidth}
                \includegraphics[width=\linewidth]{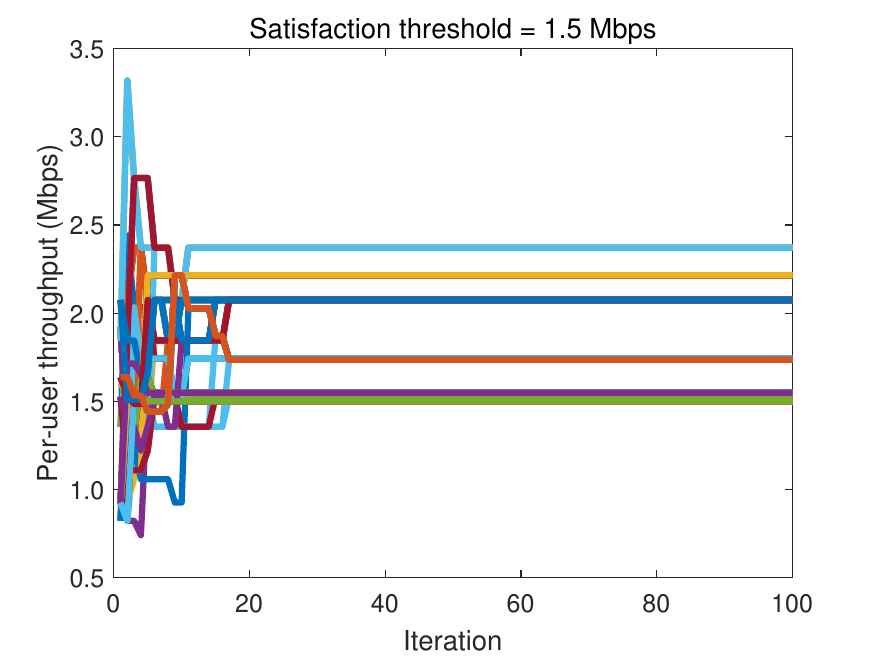}
                \caption{RM leads to a pure SE}
                \label{fig:RM_SE_1-5}
        \end{subfigure}%
        \begin{subfigure}[b]{0.39\linewidth}
                \includegraphics[width=\linewidth]{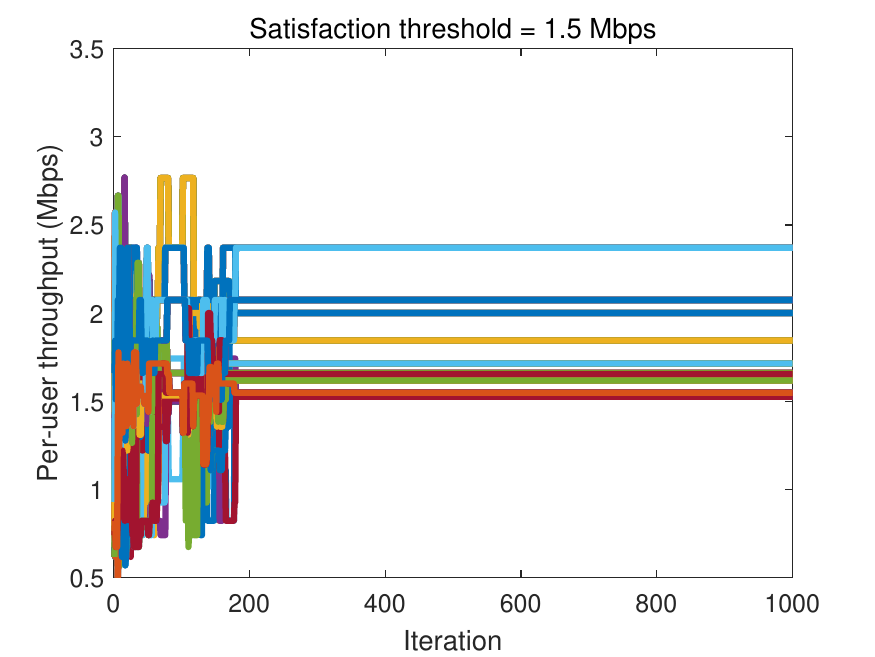}
                \caption{RMRL leads to a pure SE}
                \label{fig:RMRL_SE_1-5}
        \end{subfigure}%
        \caption{(a) Convergence towards a SE using RM (satisfied throughput = 1.5 Mbps); (b) Convergence towards a SE using RMRL (satisfied throughput = 1.5 Mbps).\label{fig:1.5Mbps}}
\end{figure*}
 \begin{figure*}[!htbp]                 
        \begin{subfigure}[b]{0.39\linewidth}
                \includegraphics[width=\linewidth]{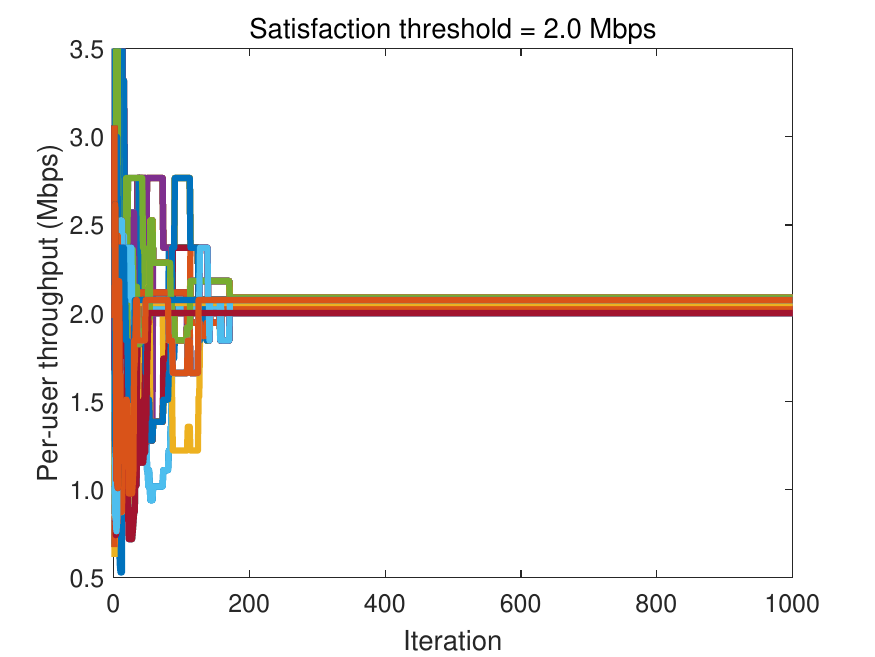}
                \caption{RM leads to a pure SE}
                \label{fig:RM_SE_2-0}
        \end{subfigure}%
        \begin{subfigure}[b]{0.39\linewidth}
                \includegraphics[width=\linewidth]{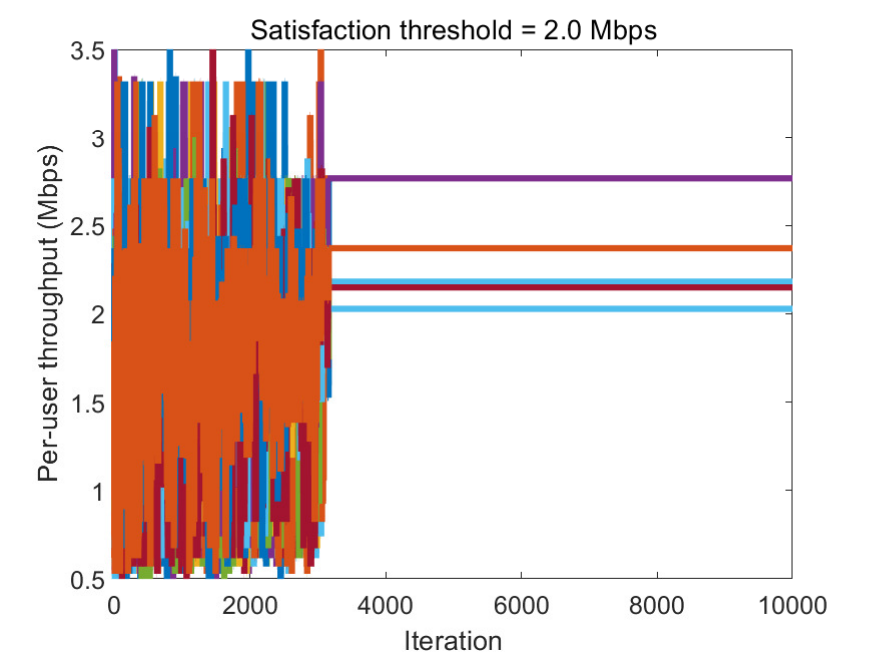}       
                \caption{RMRL leads to a pure SE}
                \label{fig:RMRL_SE_2-0}
        \end{subfigure}%
        \caption{(a) Convergence towards a SE using RM (satisfied throughput = 2.0 Mbps); (b) Convergence towards a SE using RMRL (satisfied throughput = 2.0 Mbps). \label{fig:2Mbps} }
\end{figure*}
\begin{figure*}[!htbp] 
        \begin{subfigure}[b]{0.39\linewidth}
                \includegraphics[width=\linewidth]{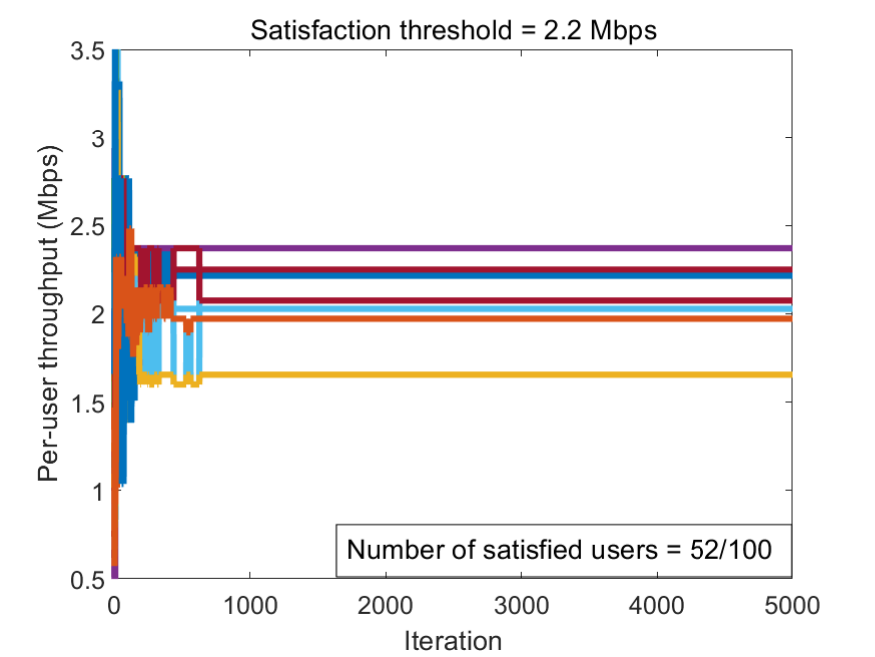}
                \caption{RM leads to a pure GSE}
                \label{fig:RM_GSE_2-2}
        \end{subfigure}%
        \begin{subfigure}[b]{0.39\linewidth}
                \includegraphics[width=\linewidth]{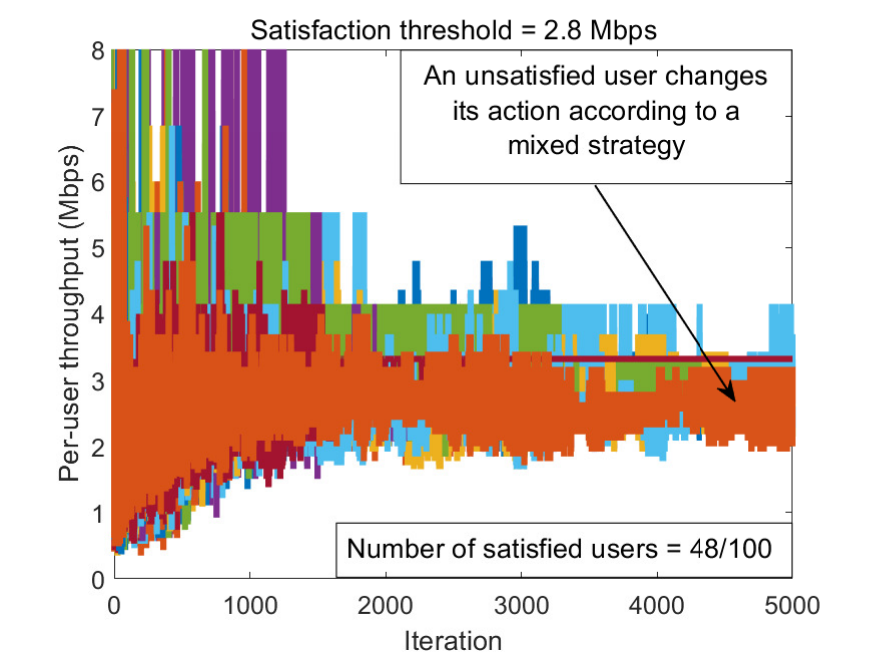}       
                \caption{RM leads to a mixed GSE}
                \label{fig:RM_MGSE_2-8}
        \end{subfigure}%
        \caption{(a) Convergence towards a pure GSE using RM (satisfied throughput = 2.2 Mbps); (b) Convergence towards a mixed GSE using RM (satisfied throughput = 2.8 Mbps).} 
        \label{fig:GSE}
\end{figure*} 
%

\section{Conclusion}\label{sec:concl}

In this technical note, we have introduced a new concept in game theory -- Bayesian satisfaction. We have considered satisfaction games (SG) where players seek to achieve a specified level of utility rather than competing to maximize their individual utilities. We consider a framework within which each player has subjective knowledge of the others' behaviors as specified by a probability distribution over all players' actions. This distribution is common knowledge. Indeed, satisfaction can be defined in a more general setting using correspondence functions that encode constraints between players' actions leading to satisfaction. This formulation is shown to lead naturally to a definition of (mixed) generalized satisfaction equilibria (GSE) for the SG that is consistent with the concept of Bayesian rationality introduced by R. Aumann (Nobel Prize 2005).\\ 

We have shown that correlated equilibria (CE) for the SG have a natural optimality property in that no player can increase its probability of satisfaction by unilateral deviation from its recommended action. The set of CE includes all GSE but also potentially many other equilibria which have not been previously studied. In one example, we compared regret-based algorithms (which find CE) with two existing algorithms from the SG literature. We found that the regret-based algorithms generally perform better in terms of efficiency of resource allocation, but, perhaps more importantly, they almost always find equilibrium outcomes whereas the existing algorithms may not. We also consider the problem of radio resource control in heterogeneous wireless networks which is an application domain where we have made previous contributions. We study the performance of the regret-based algorithms for this problem in the satisfaction framework, and observe good performance for a number of resource constrained scenarios.\\ 

The main role of the paper is to introduce the concept of Bayesian rationality in satisfaction games and to apply regret-based methods to find equilibria. Several numerical examples have been included to illustrate the potential benefits of the new algorithms proposed, but significantly more work is required to fully understand the convergence behavior of these algorithms.
 

\subsection{Future Work}

It is not the purpose of this paper to fully explore the performance of the proposed new algorithms in detail. It is clearly an important issue to further study the nature of CE in SG as well as the performance of the new algorithms proposed here, and these issues are likely to be strongly application dependent. In addition, there are examples of SG where no natural utilities are involved, and it is of interest to study these in the context of Bayesian rationality. \\

Another optimality concept in standard normal form games is called {\it Hannan consistency} \cite{hart2001general,hart2005adaptive}. In the Hannan framework, an agent reasons about whether it would be better to play any fixed action rather than one specified by the equilibrium (mixed) strategy. The Hannan set is the set of all mixed actions for which no agent can gain by choosing any fixed strategy. A {\it Hannan Equilibrium} (HE) is a pmf $\pi \in \Delta(\Ac)$ which, for every player $i$, and every action $a_i \in \Ac_i$ of player $i$, it holds that
\begin{align}
\Er_{\pi|a_i} \lbr u_i(a_i,.) \rbr & \leq \Er_\pi \lbr u_i \rbr \ .
\label{eq:HE1}
\end{align}
That is, under a HE, no player can benefit on average by playing any fixed action. For the natural satisfaction game with binary utilities \eqref{eq:HE1} becomes
\begin{align}
\Pr_{\pi|a_i} \lbr f_i^{-1}(a_i) \rbr & \leq \Pr_\pi \lbr a_i^\prime \in f_i(a_{-i}) \rbr \nonumber \\
& = \sum_{a_i^\prime \in \Ac_i} \pi_i(a_i^\prime) \, \Pr_{\pi|a_i^\prime} \lbr  f_i^{-1}(a_i^\prime) \rbr \ ,
\label{eq:HE_SE}
\end{align}
for all $i \in \Pc$ and $a_i \in \Ac_i$. It is known that the set of all HE contains all CE for the game \cite{hart2001general}. Thus the study of HE in the SG context is also of interest as more equilibria may be present, and convergence of adaptive algorithms may be better behaved. 






\bibliographystyle{ACM-Reference-Format} 
\bibliography{reference}

\newpage \clearpage
\appendix
\section{Appendix}

\subsection{An example with no utilities - satisfactory dating game}

We now consider an example where there is no natural concept of a numerical utility. We suppose there are two kinds of players $\Pc_1 = \lbr \alpha_1, \ldots, \alpha_N \rbr$ and $\Pc_2 = \lbr \beta_1, \ldots, \beta_N \rbr$. Associations are made between a player in $\Pc_1$ and one in $\Pc_2$ and are represented by the ordered pairs $\lb \alpha_i, \beta_j \rb$. For each player $\alpha_i \in \Pc_1$, there is a subset $B_i \subseteq \Pc_2$ for which the association $(\alpha_i,b), b \in B_i$ is ``satisfactory'' for player $\alpha_i$. Similarly, for each player $\beta_j \in \Pc_2$, there is a subset $A_j \subseteq \Pc_1$ for which the association $(a,\beta_j), a \in A_j$ is ``satisfactory'' for player $\beta_j$. Players $\alpha_i \in \Pc_1$ have actions $\Ac_i = \{1, \ldots, N\}$ denoting which member of $\Pc_2$ is ``selected'', i.e. player $\alpha_i$ selecting action $j$ establishes the association $(\alpha_i,\beta_j)$. Similarly for players in $\Pc_2$ which have actions $\Bc_i = \{1, \ldots, N\}$ denoting which member of $\Pc_1$ is ``selected'', i.e. player $\beta_j$ selecting action $i$ establishes the association $(\alpha_i,\beta_j)$. As usual, we define the joint actions $\Ac = \{1,\ldots, N\}^N$ and $\Bc = \{1,\ldots, N\}^N$, along with $\Ac_{-i}$ and $\Bc_{-i}$ in the usual way. A joint action profile is the pair $(a,b)$ with $a \in \Ac$, $b \in \Bc$. Write
\begin{align*}
(a,b) & = \lb n_1, \ldots, n_N, m_1, \ldots, m_N \rb \ ,
\end{align*}
i.e. players $\alpha_i$ select $\beta_{n_i}$, and players $\beta_k$ select $\alpha_{m_k}$.  
The player $\alpha_i \in \Pc_1$ is satisfied under $(a,b)$ if (i) $i \in \{m_1, \ldots, m_N\}$ (i.e. at least one of the players in $\Pc_2$ selects player $\alpha_i$), and (ii) for each player $\beta_k \in \Pc_2$ with $m_k = i$, $k \notin \{ n_1, \ldots, n_{i-1},n_{i+1}, \ldots, n_N\}$, (i.e. if a player $\beta_k$ selects $\alpha_i$, then no other player in $\Pc_1$ part from $\alpha_i$ can also select $\beta_k$. A symmetric property is imposed on players in $\Pc_2$. 

We can define the correspondence functions $f_i^1 : \Ac_{-i} \times \Bc \ra 2^{\Ac_i}$, and $f_j^2 : \Ac \times \Bc_{-j} \ra 2^{\Bc_j}$ by
\begin{align}
f_i^1 (a_{-i},b) & = \{ k : 
\beta_k \in B_i, m_k = i \} \setminus \lb \cup_{j \neq i} \{k : n_j = k \} \rb \nonumber \\
f_j^2(a,b_{-j}) & = \{ i : \alpha_i \in A_j, n_i = j \} \setminus \lb \cup_{k \neq j}  \{i : m_k = i \} \rb \ .
\end{align}
We can think of this example as a version of the well-known dating game, where there is no preference over partners, merely a notion of whether a player is satisfactory as a partner. Each player $\alpha_i$ in $\Pc_1$ is paired with a single player $\beta_j$ in $\Pc_2 \cap B_i$, and $\alpha_i \in A_j$. No player can be paired with more than one player in the other group. At a SE all players are satisfactorily paired. At a GSE, a subset of players are satisfactorily paired and a disjoint subset of players are not so paired and cannot become satisfactorily paired by unilateral deviation from the equilibrium action.


\end{document}